\documentclass[final,journal]{IEEEtran}
\IEEEoverridecommandlockouts
\IEEEpubid{\makebox[\columnwidth]{Copyright (c) 2017 IEEE.
\hfill} \hspace{\columnsep}\makebox[\columnwidth]{ }}
\usepackage{graphicx}
\usepackage{amsmath}
\usepackage{amsthm}
\usepackage{amsfonts}
\usepackage{amssymb}
\newtheorem{theorem}{Theorem}
\newtheorem*{theorem*}{Theorem}

\newtheorem{corollary}[theorem]{Corollary}

\newtheorem{proposition}[theorem]{Proposition}

\usepackage{color}
\begin{document}
\title{Quantum Dynamical Entropy, Chaotic Unitaries and Complex Hadamard Matrices}
\author{Wojciech S\l omczy\'{n}ski, Anna Szczepanek
\thanks{W. S\l omczy\'{n}ski and A. Szczepanek 
are with Institute of Mathematics, Jagiellonian University,
\L ojasiewicza 6, 30-348 Krak\'{o}w, Poland.
Email: wojciech.slomczynski@im.uj.edu.pl,
anna.szczepanek@im.uj.edu.pl}%
\thanks{Manuscript received December 13, 2016; revised September 8, 2017.}}%
\markboth{IEEE Transactions on Information Theory,~Vol.~X, No.~X, XXXX~2017}
{}
\maketitle

\begin{abstract}
We introduce two information-theoretical invariants for the projective unitary
group acting on a finite-dimensional complex Hilbert space: PVM- and
POVM-dynamical (quantum) entropies, which are analogues of the classical
Kolmogorov-Sinai entropy rate. They quantify the maximal randomness of the
successive quantum measurement results in the case where the evolution of the
system between each two consecutive measurements is described by a given
unitary operator. We study the class of chaotic unitaries, i.e., the ones of
maximal entropy, or equivalently, such that they can be represented by
suitably rescaled complex Hadamard matrices in some orthonormal bases. We
provide necessary conditions for a unitary operator to be chaotic, which
become also sufficient for qubits and qutrits. These conditions are expressed
in terms of the relation between the trace and the determinant of the
operator. We also compute the volume of the set of chaotic unitaries in
dimensions two and three, and the average PVM-dynamical entropy over the
unitary group in dimension two. We prove that this mean value behaves as the
logarithm of the dimension of the Hilbert space, which implies that the
probability that the dynamical entropy of a unitary is almost as large as
possible approaches unity as the dimension tends to infinity.
\end{abstract}


\section{Introduction}

\IEEEPARstart{I}{magine} we are standing somewhere on the Earth's (spherical)
surface that rotates around the north-south axis. Try to choose this place in
such a way as to make as large as possible the angle between the axis passing
through the chosen point and the centre of the Earth, and the rotated axis
determined after some fixed time interval. If the time period is less than six
hours, the choice is simple: we must locate ourselves somewhere on the
equator. However, if the elapsed time is chosen between six and twelve hours,
the situation becomes more complicated. We have to travel north (or south) the
equator, eventually reaching, for twelve hours, the 45th parallel north (say
on the border between Montana and Wyoming) or the 45th parallel south (e.g.,
in Becks, a small settlement on the South Island of New Zealand). In the
former `short time' case, the maximal attainable angle is equal to the earth's
angle of rotation, but in the latter, i.e., when the time is long enough, we
can always find a point on the Earth (or, more precisely, a circle of
latitude) such that the angle between the two lines in question is right. Now,
if we swap the Earth for the Bloch sphere representing qubits, this simple
riddle illustrates the difference between two kinds of unitary transformations
(represented here as Bloch sphere rotations): non-chaotic and chaotic ones.
The exploration of this difference is the main theme of this paper.

The invariant that we shall use to distinguish chaotic unitaries is quantum
dynamical entropy. The notion of (classical) \textit{dynamical entropy} (or
\textit{entropy rate}) is due to Claude E.~Shannon \cite{Sha48,ShaWea49}, who
introduced it into information theory, and Andrei Kolmogorov \cite{Kol58}, who
made it a basic tool for studying dynamical systems. In his seminal paper
Shannon discussed the problem of computing entropy for a discrete and ergodic
information source sending messages to a receiver. This quantity can be
determined from the statistics of finite message sequences, namely, it is the
limit of entropy of a block of symbols divided by its length, or the limit of
conditional entropy of the next symbol given the preceding ones, as the block
length tends to infinity. In the Kolmogorov-Sinai theory the definition of
entropy is very similar to Shannon's, except that instead of message
sequences, the results of discrete measurements (represented there by finite
partitions of the phase space) are analysed and then the supremum over all
such measurements is taken. The entropy defined in this way is invariant with
respect to metric isomorphisms of dynamical systems. Moreover, using the
Kolmogorov-Sinai (KS) entropy, we can formally distinguish \textit{regular}
systems (with dynamical entropy equal to zero) from \textit{chaotic systems}
(with strictly positive dynamical entropy).

In the present paper we consider a quantum analogue of this notion. We analyse
the situation where successive measurements are performed on a
finite-dimensional quantum mechanical system whose evolution between two
subsequent measurements is given by a quantum operation. We assume that the
dynamics of the quantum system is described by a finite-dimensional unitary
operator, and the measurement process either by a \textit{von Neumann-L\"{u}%
ders\ instrument} (represented by a projection valued measure - PVM) or by a
\textit{generalised L\"{u}ders instruments,} disturbing the initial state in
the minimal way (represented by a positive operator valued measure - POVM). If
the measure consists of rank-1 operators, then such process generates two
Markov chains: the first one in the space of states (so-called
\textit{discrete quantum trajectories}, see, e.g.,
\cite{Kum06,MaaKum06,Lim10,AttPel11,Benetal17}), and the second one in the space of
measurement outcomes. The dynamical entropy (entropy rate) of the latter can
be used to estimate the randomness of the measurement results.

Understood in this way, quantum dynamical entropy was introduced independently
by Srinivas \cite{Sri78}, Pechukas \cite{Pec82}, Beck and Graudenz
\cite{BecGra92}, cf. \cite[Sec. 4.2]{AccOhyWat97} and \cite{KolKon14}.
The idea has been recently rediscovered and
analysed by Crutchfield and Wiesner under the name of \textit{quantum entropy
rate} \cite{CruWie08,Wie10}. They have provided also a detailed information-theoretic
interpretation of this notion, as well as of two related notions: \textit{excess entropy}
and \textit{transient information}, see also \cite{CruFel03}. The entropy rate says
how predictable the measurement results are, the excess entropy - how
hard it is to do the predicting, and the transient information - how difficult it is to
know the internal state of such a quantum process through measurements.
The notion of entropy rate is also closely related to the \textit{entropy of unitary matrices,}
used in different contexts by various authors \cite{Knoetal99,Zycetal03,Ail16}.

Imitating the definition of the
Kolmogorov-Sinai entropy \cite[p. 64]{Gre11} and taking the supremum over the class of PVM
measurements, we get the \textit{PVM-dynamical entropy}, which depends only on
the quantum dynamics and characterizes its ability to produce random sequences
of measurement outcomes. In the case of POVMs the situation is more
complicated as there are two independent sources of randomness that can
influence the value of dynamical entropy. The first is the underlying dynamics
of the system, described by a unitary operator. The second is the POVM
measurement, which potentially introduces some additional randomness.
Subtracting the dynamical entropy calculated for trivial (identity) dynamics
from the original entropy rate, and then taking the supremum over the class of
POVM measurements, we get another quantity, the \textit{POVM-dynamical
entropy}, which again depends only on the unitary operator and is larger than
or equal to its PVM counterpart. These measurement independent definitions of
dynamical quantum entropy for finite-dimensional systems were introduced in a
more general setting in \cite{SloZyc94,Kwaetal97,SloZyc98} and then developed
further in \cite{Slo03,SloSzy16}. However, only preliminary results have been
obtained so far. In the present paper we study the notion of PVM-dynamical
entropy in full details, postponing more comprehensive analysis of the
POVM-dynamical entropy to further publications. The PVM-dynamical entropy
quantifies the maximal rate at which classical randomness can be produced by
a given unitary dynamics in a repeated von Neumann-L\"{u}ders measurement process.
In this sense, it is a natural counterpart of the classical KS entropy
modelled on the notion of entropy of an information channel.

Note that a widely accepted generalization of the KS entropy for quantum mechanics
has not yet been found, in spite of the fact that several attempts to define
such a quantity have been made \cite{OhyPet93,Cap05}. In particular, the
best-known quantum dynamical entropies, such as the Connes-Narnhofer-Thirring
(CNT) entropy \cite{Conetal87} or the Alicki-Fannes (AF) \cite{AliFan01}
entropy, vanish for finite-dimensional quantum systems \cite{Benetal98},
\cite[Sec. 14.5]{Ben03}, and so they cannot be used to quantify the randomness
of the successive measurement outcomes in the case we study here.

In Sec.~\ref{BASIC} we introduce the notions of PVM- and POVM-dynamical
entropy and observe that they are invariant under conjugation, inversion and
phase multiplication, which makes them class functions for the projective
unitary-antiunitary group. These quantities are non-negative and bounded from
above by the logarithm of the dimension of the underlying Hilbert space, also
called the number of degrees of freedom of the quantum mechanical system. We
show that their mean values averaged over the unitary ensemble are only
slightly smaller than the upper bound and so tend logarithmically to infinity.
In Sec.~\ref{CHAOS} we use these dynamical entropies to distinguish between
\textit{chaotic}, i.e., the ones of maximal entropy, and non-chaotic
unitaries. The former are characterized as those that can be represented by a
suitably rescaled complex Hadamard matrix in some orthonormal basis. In
Sec.~\ref{QUBITS} we compute the volume of the set of chaotic matrices as well
as the exact value of mean PVM-dynamical entropy in dimension two.
Sec.~\ref{QUTRITS} contains a necessary condition for a unitary matrix to be
chaotic. We show that for qubits and qutrits this condition is, in fact,
sufficient. This allows us to compute the volume of the set of chaotic
matrices also in dimension three. In Sec.~\ref{POVM} we discuss the
difficulties that arise when trying to extend the definition of quantum
dynamical entropy to the realm of general measurements.

\section{Quantum dynamical entropy - definition and basic
properties\label{BASIC}}

We assume that the pure states of a $d$-dimensional quantum system are
represented by the complex projective space $\mathbb{CP}^{d-1}$ or,
equivalently, by the set $\mathcal{P}\left(  \mathbb{C}^{d}\right)  $ of
one-dimensional projections in $\mathbb{C}^{d}$. The set of all quantum states
$\mathcal{S}\left(  \mathbb{C}^{d}\right)  $ is the convex closure of
$\mathcal{P}\left(  \mathbb{C}^{d}\right)  $, i.e., the set of density
(Hermitian, positive semi-definite, and trace one) operators on $\mathbb{C}%
^{d}$.

The \textsl{measurement }(with $k$ possible outcomes) of this system is given
by a \textsl{positive operator valued measure} (\textsl{POVM}), i.e., an
ensemble of positive (non-zero) Hermitian operators $\Pi_{j}$ ($j=1,\ldots,k$)
on $\mathbb{C}^{d}$ that sum to the identity operator, i.e., $\sum
\nolimits_{j=1}^{k}\Pi_{j}=\mathbb{I}$. In this paper we shall consider
only \textsl{normalized rank-}$1$\textsl{\ POVMs}, where $\Pi_{j}$
($j=1,\ldots,k$) are \mbox{rank-$1$} operators and $\operatorname{tr}\left(  \Pi
_{j}\right)  =\operatorname{const}(j)=d/k$, but we shall discuss shortly
the general case in the last section. Necessarily, $k\geq d$
and there exists an ensemble of pure states $\left|  \varphi
_{j}\right\rangle \left\langle \varphi_{j}\right|  \in\mathcal{P}\left(
\mathbb{C}^{d}\right)  $ ($j=1,\ldots,k$) such that $\Pi_{j}=\left(
d/k\right)  \left|  \varphi_{j}\right\rangle \left\langle \varphi_{j}\right|
$. (Here and henceforth, we use Dirac's bra-ket notation.) Thus,
$\sum\nolimits_{j=1}^{k}\left|  \varphi_{j}\right\rangle \left\langle
\varphi_{j}\right|  =\left(  k/d\right)  \cdot\mathbb{I}$. In particular, if
$k=d$ and so $\left(  \varphi_{j}\right)  _{j=1}^{d}$ is an orthonormal basis
of $\mathbb{C}^{d}$, we get a special class of \textsl{rank-}$1$\textsl{ projection valued
measures} (\textsl{PVMs}).

If the state of the system before the measurement (the \textsl{input state})
is $\rho\in\mathcal{S}\left(  \mathbb{C}^{d}\right)  $, then the probability
$p_{j}\left(  \Pi,\rho\right)  $ of the $j$-th outcome is given by
$p_{j}\left(  \Pi,\rho\right)  :=\operatorname{tr}\left(  \rho\Pi_{j}\right)
$ for $j=1,\ldots,k$. In particular, for normalized rank-$1$ POVMs we have
$p_{j}\left(  \Pi,\rho\right)  =\left(  d/k\right)  \left\langle \varphi
_{j}\right|  \rho\left|  \varphi_{j}\right\rangle $, and if ${\rho=\left|
\psi\right\rangle \left\langle \psi\right|  \in\mathcal{P}\left(
\mathbb{C}^{d}\right)  }$, then $p_{j}\left(  \Pi,\rho\right)  =\left(
d/k\right)  \left|  \left\langle \varphi_{j}|\psi\right\rangle \right|  ^{2}$
(the \textsl{Born rule}). The measurement process generically alters the state
of the system, but the POVM alone is not sufficient to determine the
post-measurement (or \textsl{output}) state. This can be done by defining a
\textsl{measurement instrument} (in the sense of Davies and Lewis
\cite{DavLew70}) compatible with $\Pi$, see also \cite[Ch. 5]{HeiZim11}. We
shall only consider here the so-called \textsl{generalised L\"{u}ders
instruments}, disturbing the initial state in the minimal way, see
\cite[p. 404]{DecGra07}, where the output state is $\left|  \varphi
_{j}\right\rangle \left\langle \varphi_{j}\right|  $, providing the result of
the measurement was $j$.

Consider the situation where the successive measurements described by $\Pi$ are performed
on an evolving quantum system. We assume that the motion of the system between two subsequent
measurements is governed by ${U\in\mathsf{U}\left(  d\right)}  $ acting as
$\mathcal{S}\left(  \mathbb{C}^{d}\right)  \ni\rho\rightarrow U\rho\,U^{\ast
}\in\mathcal{S}\left(  \mathbb{C}^{d}\right)  $. Then the results of
consecutive measurements are represented by finite strings of letters from a
$k$-element alphabet. The probability of obtaining the string $\left(
i_{1},\ldots,i_{n}\right)  $, where $i_{m}=1,\ldots,k$ for $m=1,\ldots,n$, and
$n\in\mathbb{N}$, is then given by the \textit{generalized Wigner formula} \cite{Wig63}
\[
P_{i_{1},\ldots,i_{n}}\left(  \rho\right)  :=p_{i_{1}}\left(  \rho\right)
\cdot
{\textstyle\prod\nolimits_{m=1}^{n-1}}
p_{i_{m}i_{m+1}}\text{,}%
\]
where $\rho$ is the initial state of the system, $p_{j}\left(  \rho\right)
:=\left(  d/k\right)  \left\langle \varphi_{j}\right|  \rho\left|  \varphi
_{j}\right\rangle $ is the probability of obtaining $j$ in the first
measurement, and $p_{jl}:=\left(  d/k\right)  \left|  \left\langle \varphi
_{j}|U|\varphi_{l}\right\rangle \right|  ^{2}$ is the probability of getting
$l$ as the result of the measurement, providing the result of the preceding
measurement was $j$, for $j,l=1,\ldots,k$. In consequence, the combined
evolution of states is Markovian with the initial distribution given by
$p:=(p_{j})_{j=1}^{k}$ and the transition matrix $P:=(p_{jl})_{j,l=1}^{k}$
\cite{SloZyc94,Slo03}. For rank-$1$ PVMs this matrix is \textsl{unistochastic}.

The randomness of the measurement outcomes can be analysed with the help of
\textsl{quantum entropy of} $U$ \textsl{with respect to }$\Pi$ defined in a way
analogous to the Kolmogorov-Sinai entropy, namely
{\
\begin{align}
H(U,\Pi):=\lim_{n\rightarrow\infty}\left(  H_{n+1}-H_{n}\right)
=\lim_{n\rightarrow\infty}\frac{H_{n}}{n}\text{,}%
\label{limits}
\end{align}
}
where
\[
H_{n}:=\sum_{j_{1},\ldots,j_{n}=1}^{k}\eta\left(  P_{j_{1},\ldots,j_{n}%
}\left(  \rho_{\ast}\right)  \right)
\]
with the \textsl{Shannon function} $\eta\colon\mathbb{R}^{+}\rightarrow
\mathbb{R}$ defined by $\eta(x):=-x\ln x$ for $x>0$ and $\eta(0):=0$, where
$\rho_{\ast}:=\mathbb{I}/d$. It is easy to check that both limits in
(\ref{limits}) exist and are equal. The maximally mixed state $\rho_{\ast}$
plays here the role of the `stationary state' for Markov evolution with
$p_{j}\left(  \rho_{\ast}\right)=1/k$ for $j=1,\ldots,k$ \cite{SloZyc94,Slo03}.
It represents an \textsl{unprepared} quantum system.

Using the formula for the entropy of a Markov chain, which is a special case
of a much more general \textsl{integral entropy formula} \cite{Slo03}, it is
easy to show \cite[eq. (24)]{SloSzy16} that
\begin{align}
H(U,\Pi)  &  =\frac{1}{k}\sum_{j,l=1}^{k}\eta\left(  p_{jl}\right)  \nonumber \\ & =\ln\frac{k}{d}+\frac{d}{k^{2}}\sum_{j,l=1}^{k}\eta\left(  |\left\langle
\varphi_{j}\right|  U\left|  \varphi_{l}\right\rangle |^{2}\right)  \text{.}
\label{formula}
\end{align}
In consequence,
\begin{equation}
\ln\left(  k/d\right)  \leq H(U,\Pi)\leq\ln k\text{.} \label{est}%
\end{equation}

There are two possible sources of randomness in this model, the measurement process and
the underlying unitary dynamics, and we would like to quantify their impact
separately. This can be done by defining two quantities:

\begin{itemize}
\item \emph{\ }the \textsl{measurement entropy of }$\Pi$ given by
\[
H_{\operatorname*{meas}}(\Pi):=H(\mathbb{I},\Pi)\text{;}%
\]

\item \emph{\ }the \textsl{dynamical entropy of }$U$\textsl{\ with respect to
}$\Pi$ given by
\[
H_{\operatorname*{dyn}}(U,\Pi):=H(U,\Pi)-H_{\operatorname*{meas}}(\Pi)\text{.}%
\]
\end{itemize}

Now, we introduce two kinds of measurement independent quantum dynamical
entropies, by maximizing $H_{\operatorname*{dyn}}$ either over all PVMs or all
POVMs. Namely, the \textsl{PVM-dynamical entropy~of}~$U$:%
\begin{equation}
H^{\operatorname*{dyn}}(U):=\max_{\Pi\in\mathsf{PVM}}H_{\operatorname*{dyn}}(U,\Pi)\text{,}
\label{entPVM}%
\end{equation}
and the \textsl{POVM-dynamical entropy of }$U$:%
\begin{equation}
\overline{H}^{\operatorname*{dyn}}(U):=\sup_{\Pi\in\mathsf{POVM}}H_{\operatorname*{dyn}}%
(U,\Pi)\text{.} \label{entPOVM}%
\end{equation}

Analogously, we can define quantum dynamical entropy for an antiunitary transformation.

Note that the set of all PVMs, i.e., all projective orthonormal (ordered)
bases, forms, endowed with a natural topology, a compact space isomorphic
to the $d(d-1)$-dimensional flag manifold $\mathsf{U}(d)/\mathsf{U}(1)^{d}$ \cite[p. 133]{BenZyc06}.
Now, it follows from (\ref{formula}) that $H_{\operatorname*{dyn}}$ is continuous
in both variables. Hence, the supremum is attainable in (\ref{entPVM}) and $H^{\operatorname*{dyn}}$ is continuous.

For every $U\in\mathsf{U}\left(  d\right)  $ we have $\min_{\Pi
\in\mathsf{PVM}}H(U,\Pi)=0$, since the PVM $\Pi$ generated with the help of an
eigenbasis of $U$ gives $H(U,\Pi)=0$. In consequence, we cannot define here a
quantum counterpart of classical \textit{Kolmogorov automorphisms}
(\textit{K-systems}), i.e., maps with positive entropy with respect to all
non-trivial finite partitions of the phase space.

For $\Pi\in\mathsf{PVM}$ we have $H_{\operatorname*{meas}}(\Pi)=0$, and so
$H_{\operatorname*{dyn}}(U,\Pi)=H(U,\Pi)$. Consequently, we get
\[
H^{\operatorname*{dyn}}(U)=\max_{\Pi\in\mathsf{PVM}}H(U,\Pi)\text{,}%
\]
which implies
\[
H^{\operatorname*{dyn}}(U)=\max_{(e_{j})_{j=1}^{d}}\frac{1}{d}\sum_{j,l=1}%
^{d}\eta\left(  |\left\langle e_{j}\right|  U\left|  e_{l}\right\rangle
|^{2}\right)  \text{,}%
\]
where the maximum is taken over all orthonormal bases. Equivalently, we can
fix a basis (e.g., an eigenbasis of $U$) and take the maximum over all unitary
transformations:
\begin{equation}
H^{\operatorname*{dyn}}(U)=\max_{V\in\mathsf{U}(d)}\frac{1}{d}\sum_{j,l=1}%
^{d}\eta(|\left(  V^{\ast}UV\right)  _{jl}|^{2})\text{.} \label{entdynV}
\smallskip
\end{equation}

Moreover, from (\ref{est}) we get $\left|  H_{\operatorname*{dyn}}\left(
U,\Pi\right)  \right|  \leq\ln d$ and
\[
0\leq H^{\operatorname*{dyn}}(U)\leq\overline{H}^{\operatorname*{dyn}}%
(U)\leq\ln d\text{.}%
\]
The bounds are achievable, as we have $\overline{H}^{\operatorname*{dyn}%
}(\mathbb{I})=0$ and $H^{\operatorname*{dyn}}(F_{d}/\sqrt{d})=\ln d$, where
$F_{d}/\sqrt{d}$ is a unitary operator called the \textsl{quantum Fourier transform},
with $F_{d}$ represented in some basis by the \textsl{Fourier matrix} of size $d$,
given by $(\omega_{d}^{\left(  j-1\right)  \left(l-1\right)  })_{j,l=1}^{d}$
with $\omega_{d}:=\exp(2\pi i/d)$.

The following proposition that summarizes facts concerning invariance of the
dynamical entropies is easy to show.

\begin{proposition}
[invariance]The dynamical entropies ${H}^{\operatorname*{dyn}}$ and
$\overline{H}^{\operatorname*{dyn}}$ are invariant under the following operations

\begin{enumerate}
\item [(i)]\textsl{conjugation}: $U\rightarrow V^{-1}UV$ for every unitary or
antiunitary~$V$;

\item[(ii)] \textsl{inversion}: $U\rightarrow U^{-1}$;

\item[(iii)] \textsl{phase multiplication}: $U\rightarrow e^{i\varphi}U$ for
$\varphi\in\mathbb{R}$.
\end{enumerate}
\end{proposition}

It follows from (i) above that both these quantities are unitary
(and antiunitary) invariants (i.e., unitary class functions), and so they
depend only on the spectrum of $U$, since two unitary matrices are unitarily
similar if and only if they have the same spectrum (treated as a multiset).
Moreover, (ii) implies that they are time-reversal invariants. According
to (iii), both quantum dynamical entropies are also
projective invariants, and so they can be treated as class functions for the
projective unitary-antiunitary group. Let us now see how these facts can be used
to characterize the domain of both entropies.

First, from the above considerations it follows that the space of
conjugacy classes of unitary matrices is isomorphic to the $d$-th
symmetric product of $S^{1}$, i.e., the space of \mbox{$d$-element} multisets
contained in $S^{1}$, denoted by $SP^{d}(S^{1})$. Morton proved
that $SP^{d}(S^{1})$ is a fibre bundle over $S^{1}$ and the fibres are
$(d-1)$-dimensional discs \cite{Mor67}. Moreover, he showed that the bundle is trivial if
$d$ is odd, and it is non-orientable if $d$ is even, e.g.,
$SP^{1}(S^{1}) \simeq S^{1}$ and $SP^{2}(S^{1})$ is the M\"{o}bius strip.
Taking into account the phase multiplication invariance, one may further reduce
the domain of dynamical entropies to a set topologically isomorphic to
the $(d-1)$-dimensional disc. In particular, we show in Sec.~\ref{QUBITS} and
Sec.~\ref{QUTRITS}, respectively, that for $d=2$ the value of
${H}^{\operatorname*{dyn}}(U)$ depends on one real parameter, the angle
between two eigenvalues of $U$, and for $d=3$ it is a function of one
complex parameter, the trace of $U$ divided by a cube root of its
determinant.

To lower bound the mean value of the PVM-dynamical entropy averaged over the
ensemble of unitary matrices, we consider yet another unitary invariant, the
\textsl{PVM-average dynamical entropy}, given by $M(U):=\left\langle
H(U,\Pi)\right\rangle _{\Pi\in\mathsf{PVM}}$ for $U\in\mathsf{U}(d)$. Namely,
we have

\begin{theorem}
[mean entropy bounds]\label{meaent}
\[
\ln d-(1-\gamma)<\sum_{k=2}^{d}\frac{1}{k}=\langle M(U)\rangle
_{\mathsf{U}(d)}\negthinspace<\negthinspace\langle H^{\operatorname*{dyn}}(U)\rangle
_{\mathsf{U}(d)}\negthinspace <\negthinspace \ln d\text{,}%
\]
where $\gamma\approx0.577$ is Euler's constant.
\end{theorem}

\begin{proof}
All entropies are bounded from above by $\ln d$. On the other hand, from Jones
(\mbox{\cite[eq. (13)]{Jon90}} and \cite[eq. (27)]{Jon91}), see also
\cite{SloZyc98,Zycetal03}, we deduce that $\left\langle H(U,\Pi)\right\rangle
_{\mathsf{U}(d)}=\sum_{k=2}^{d}\frac{1}{k}$ for every $\Pi\in\mathsf{PVM}$.
Hence,
\begin{align*}
\gamma-1+\ln d  &  <\sum_{k=2}^{d}\frac{1}{k}=\left\langle M(U)\right\rangle
_{\mathsf{U}(d)}\\
&  =\max_{\Pi\in\mathsf{PVM}}\left\langle H(U,\Pi)\right\rangle _{\mathsf{U}%
(d)} \\
& \leq\left\langle \max_{\Pi\in\mathsf{PVM}}H(U,\Pi)\right\rangle
_{\mathsf{U}(d)}\\
&  =\left\langle H^{\operatorname*{dyn}}(U)\right\rangle _{\mathsf{U}(d)}%
\leq\ln d\text{.}%
\end{align*}
From the continuity of $H^{\operatorname*{dyn}}$ it follows that the last two
inequalities are strict.
\end{proof}

In consequence, we see that the mean values of both entropies
$H^{\operatorname*{dyn}}$ and $\overline{H}^{\operatorname*{dyn}}$ are almost
as large as possible and increase logarithmically with the dimension of the
Hilbert space. Moreover, from Chebyshev's inequality we deduce that the
probability of $H^{\operatorname*{dyn}}\leq\ln d-f(d)$,
for $f:\mathbb{N}\rightarrow\mathbb{R}^{+}$, is smaller than
$(1- \gamma)/f(d)$, and so it tends to $0$, providing
$f(d)\rightarrow\infty$, even if the latter convergence is very slow. In Sec.
\ref{QUBITS} we compute the exact value of $\left\langle
H^{\operatorname*{dyn}}(U)\right\rangle _{\mathsf{U}(d)}$ for $d=2$.

\section{Entropy-maximising unitaries\label{CHAOS}}

The concept of quantum dynamical entropy specifies a special class of
entropy-maximising unitaries, such as the Fourier quantum transforms mentioned above.
We shall call them \textsl{chaotic} since they can be used to produce
maximally random sequences of measurement results. As we shall see, this
property does not depend on which of the two definitions we work with. Namely,
from (\ref{est}) it follows that
\[
H_{\operatorname*{dyn}}(U,\Pi)=\ln d \ \ \text{iff}
\]
\begin{equation}
H(U,\Pi)=\ln k \ \ \text{and}\ \ H(\mathbb{I},\Pi)=\ln(k/d)\text{.} \label{max}
\end{equation}
By (\ref{formula}), we get $H(\mathbb{I},\Pi)=\ln(k/d)$ if and only if $\Pi$
is a PVM, i.e., $k=d$, and then, clearly, $H(\mathbb{I},\Pi)=0$. Thus,
\[
\overline{H}^{\operatorname*{dyn}}(U)=\ln d\ \ \text{iff}%
\ \ H^{\operatorname*{dyn}}(U)=\ln d\text{.}%
\]
Moreover, chaotic unitaries turn out to be exactly those that are represented
by a suitably rescaled \textsl{complex Hadamard matrix} in some basis.

\begin{proposition}
\label{Hadamard}Let $U\in\mathsf{U}(d)$. Then the following conditions are equivalent:

\begin{enumerate}
\item [(i)]$U$ is chaotic;

\item[(ii)] there exists an orthonormal basis $\{e_{j}\}_{j=1}^{d}$ such that
$\frac{1}{d}\sum_{j,l=1}^{d}\eta\left(  |\left\langle e_{j}\right|  U\left|
e_{l}\right\rangle |^{2}\right)  =\ln{d}$;

\item[(iii)] there exists an orthonormal basis $\{e_{j}\}_{j=1}^{d}$ such that
$\{e_{j}\}_{j=1}^{d}$ and $\{Ue_{j}\}_{j=1}^{d}$ are mutually unbiased;

\item[(iv)] there exists an orthonormal basis $\{e_{j}\}_{j=1}^{d}$ such that
$\sum_{j,l=1}^{d}|\left\langle e_{j}\right|  U\left|  e_{l}\right\rangle
|=d\sqrt{d}$;

\item[(v)] $\sqrt{d}\,U$ is represented by a complex Hadamard matrix in some
orthonormal basis $\{e_{j}\}_{j=1}^{d}$, i.e., $|\left\langle e_{j}\right|
U\left|  e_{l}\right\rangle |=1/\sqrt{d}$ for each $j,l=1,\ldots,d$.
\end{enumerate}
\end{proposition}

\begin{proof}
The equivalence of (i) and (ii) follows immediately from (\ref{formula}) and
(\ref{max}). As the Shannon entropy is maximal only for the uniform
distribution, all expressions of the form $|\left\langle e_{j}\right|  U\left|
e_{l}\right\rangle |^{2}$ for $j,l=1,\ldots,d$ must be equal, which proves the
equivalence of (ii) and (v). On the other hand, (v) is just (iii) expressed in
another way. The equivalence of (iv) and (v) follows from \cite[Proposition 4.12]{Ban12}.
\end{proof}

The fact that Hadamard matrices saturate the upper bound for the so-called
entropy of a unitary matrix is well known \cite{Zycetal03}. Observe, however,
that the analogous problem for real orthogonal matrices is highly non-trivial
\cite{Gadetal03,Pat12}, since real Hadamard matrices can exist only if $d=1,2$ or is
a multiple of $4$.

From Proposition \ref{Hadamard} we deduce immediately a simple necessary
condition for $U$ to be chaotic.

\begin{corollary}
\label{trace}If $U\in\mathsf{U}(d)$ is chaotic, then $\left|
\operatorname*{tr}U\right|  \leq\sqrt{d}$.
\end{corollary}

We shall see in Sec. \ref{QUTRITS} that for $d=2$, in contrast to higher
dimensions, this condition is also sufficient.

As \textsl{quantum gates} are represented by unitaries (defined up to a phase)
we can talk about \textsl{dynamical entropies of quantum gates} and we can
distinguish the class of \textsl{chaotic quantum gates}. Using the formula for
dynamical entropy presented in the next section, we shall see that among
chaotic unitaries one can list many well-known quantum gates, including the
Hadamard, \textrm{NOT} (Pauli-$X$), Pauli-$Y$, Phase Flip (Pauli-$Z$),
$\pi/4$-phase shift and $\sqrt{\text{\textrm{NOT}}}$ gates in dimension two. Also
the \textrm{CNOT} (\textrm{XOR}), \textrm{CSIGN}, \textrm{SWAP}, and
\textrm{iSWAP} gates in dimension four belong to this class.
To see this, observe that all these matrices are
unitarily similar to $D:=\operatorname*{diag}\left(  1,1,1,-1\right)  $.
The spectra of $D$ and the real-valued Hadamard matrix $F_{4}^{(1)}(3\pi/2)$,
see \cite{TadZyc06}, coincide. Thus, by Proposition~\ref{Hadamard}, $D$ is
chaotic. On the other hand, the $\pi/8$-phase shift gate in dimension two, as
well as the $\sqrt{\text{\textrm{CNOT}}}$ and $\sqrt{\text{\textrm{SWAP}}}$
gates in dimension four, do not fulfill the trace condition from Corollary
\ref{trace}, and so they are not chaotic. It follows also from this corollary
that among controlled-$U$ gates in dimension four, only the ones with
$U$ equivalent up to conjugation and phase multiplication to \textrm{NOT},
like \textrm{CNOT} or \textrm{CSIGN}, are chaotic.
In the same way we argue that multiqubit controlled
gates, like Toffoli (\textrm{CCNOT}), Fredkin (\textrm{CSWAP}) or Deutsch
(\textrm{CCR}) gates in dimension eight, cannot be chaotic.

\section{PVM-dynamical entropy: qubits\label{QUBITS}}

Computing the PVM-dynamical entropy in dimension two is a relatively easy
task, as the optimization problem reduces to finding the maxima of real-valued
functions belonging to a one-parameter family. Here, the parameter is the
angle between two eigenvalues of a unitary map. The formula for entropy in
this case has been already obtained in \cite{Slo03}, but for the sake of
completeness we recall hereafter its proof.

Let $U\in\mathsf{U}\left(  \mathbb{C}^{2}\right)  $ with the spectrum
$\left\{  \exp\left(  i\varphi\right)  ,\exp\left(  i\psi\right)  \right\}  $,
where $\varphi,\psi\in\left[  0,2\pi\right)  $. Fix an eigenbasis of $U$. In
this basis $U$ is represented by the matrix
\[
U\sim\left[
\begin{array}
[c]{cc}%
\exp\left(  i\varphi\right)   & 0\\
0 & \exp\left(  i\psi\right)
\end{array}
\right]  \text{.}%
\]
Consider now $V\in\mathsf{U}\left(  \mathbb{C}^{2}\right)  $ given by
\[
V\sim\left[
\begin{array}
[c]{cc}%
u & v\\
w & z
\end{array}
\right]  \text{,}%
\]
where $u,v,w,z\in\mathbb{C}$ satisfy $\left|  u\right|  ^{2}+\left|  v\right|
^{2}=\left|  w\right|  ^{2}+\left|  z\right|  ^{2}=1$ and $u\overline
{w}+v\overline{z}=0$. Then
\[
V^{\ast}UV\sim\left[
\begin{array}
[c]{cc}%
\left|  u\right|  ^{2}e^{i\varphi}+\left|  w\right|  ^{2}e^{i\psi} &
v\overline{u}e^{i\varphi}+z\overline{w}e^{i\psi}\\
u\overline{v}e^{i\varphi}+w\overline{z}e^{i\psi} & \left|  v\right|
^{2}e^{i\varphi}+\left|  z\right|  ^{2}e^{i\psi}%
\end{array}
\right]  \text{.}%
\]
Put $p:=\left|  u\right|  ^{2}\in\left[  0,1\right]  $, $\theta:=
\min \left( \left| \varphi -\psi \right| ,2\pi -\left| \varphi -\psi \right| \right)\in
\left[0,\pi\right]$,
and $c:=\sin^{2}\left(\theta/2\right)  \in\left[  0,1\right]  $.
As $\left|  z\right|  ^{2}=p$ and $\left|  w\right|  ^{2}=\left|  v\right|  ^{2}=1-p$,
we obtain
\begin{align}
& \frac{1}{2}\sum_{j,l=1}^{2}\eta(|\left(  V^{\ast}UV\right)  _{jl}|^{2}) \nonumber \\
& =\eta\left(  4p\left(  1-p\right)  c\right)  +\eta\left(  1-4p\left(
1-p\right)  c\right)  \text{.}\label{h_c}%
\end{align}
Denote the right-hand side of (\ref{h_c}) by $h_{c}\left(  p\right)  $. Then
$h_{c}\colon\left[  0,1\right]  \rightarrow\mathbb{R}$ attains the maximum
equal to $\ln2$ at $\frac{1}{2}(1\pm\sqrt{1-(2c)^{-1}})$ for  $c\geq1/2$, and
equal to $\eta\left(  c\right)  +\eta\left(  1-c\right)  $ at $1/2$ for
$c\leq1/2$. Using this fact and (\ref{entdynV}), we obtain (see Fig. 1)

\begin{proposition}
\label{entdyn2}%
\begin{equation}
H^{\operatorname*{dyn}}(U)\!=\!\left\{  \!%
\begin{array}
[c]{cc}%
\ln2 & \theta\geq\frac{\pi}{2}\\
\eta\left(  \cos^{2}\left(  \frac{\theta}{2}\right)  \right)  +\eta\left(
\sin^{2}\left(  \frac{\theta}{2}\right)  \right)  & \theta\leq\frac{\pi}{2}%
\end{array}
\right.  \text{\thinspace.} \label{fordim2}%
\end{equation}
\end{proposition}

\begin{center}
\includegraphics[scale=0.4]{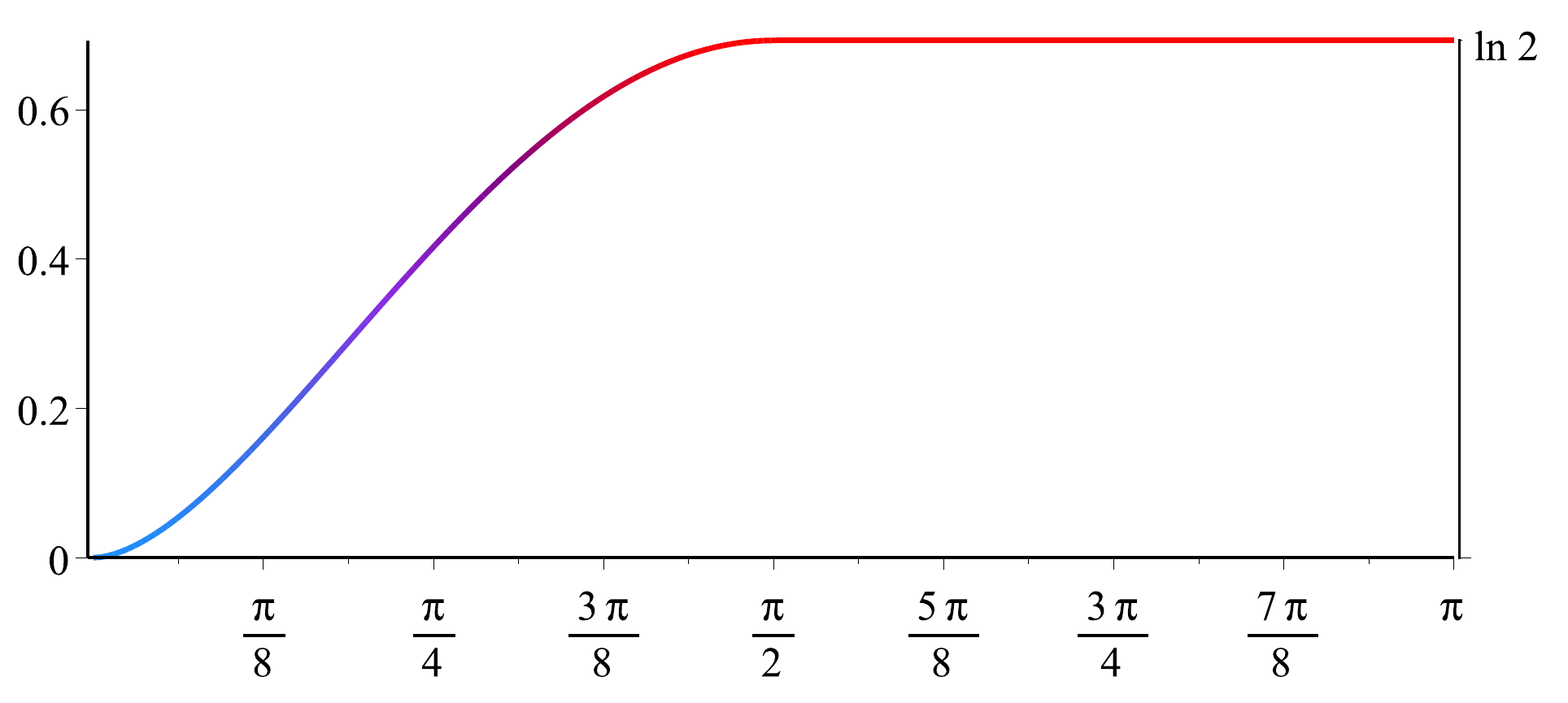}

Fig. 1. $H^{\operatorname*{dyn}}$\ as a function of $\theta$ (the chaotic
part in \textit{red}).
\end{center}

\noindent Denote by $\left\{  \left|  0\right\rangle ,\left|  1\right\rangle
\right\}  $ the eigenbasis of $U$. Observe that

\begin{itemize}
\item  the critical point at which $H^{\operatorname*{dyn}}$ hits its maximum
possible value $\ln2$ is $\theta=\frac{\pi}{2}$; this applies to well-known
$\pi/4$-phase shift and $\sqrt{\text{NOT}}$ gates;

\item  the PVMs with respect to which $H(U,\Pi)$ attains its maximal value are
given by the bases $\left\{  \left|  x^{\tau}\right\rangle ,\left|  x_{\perp
}^{\tau}\right\rangle \right\}  $, where $\tau$ is an arbitrary number from
$[0,2\pi)$, $\left|  x^{\tau
}\right\rangle :=\sqrt{r}\left|  0\right\rangle +e^{i\tau}\sqrt{1-r}\left|
1\right\rangle $ with $r:=\frac 1 2$ for $\theta\leq\frac{\pi}{2}$ and  $r:=\frac{1}{2}(1\pm\sqrt{1-(2\sin^{2}(\theta/2))^{-1}%
})$ for $\theta\geq\frac{\pi}{2}$.
\end{itemize}

The geometric interpretation of the latter fact, mentioned already in the
introduction, is the following. Fix the Bloch vectors corresponding to the
eigenbasis of $U$ as the north and south poles of the Bloch sphere. Then $U$
can be interpreted as the rotation around the north-south axis by the angle
\mbox{$\theta$}. Under this picture, finding a maximizing PVM\ is
equivalent to choosing the appropriate axis such that the angle between this
axis and its image under the rotation is maximal. If $\theta$ is acute, then
the axis must lie in the equatorial plane and the angle in question is equal
to $\theta$, but if $\theta$ is obtuse, we can find an axis that can be
transformed into a perpendicular thereto by the rotation.

\begin{center}
a) \includegraphics[scale=0.35]{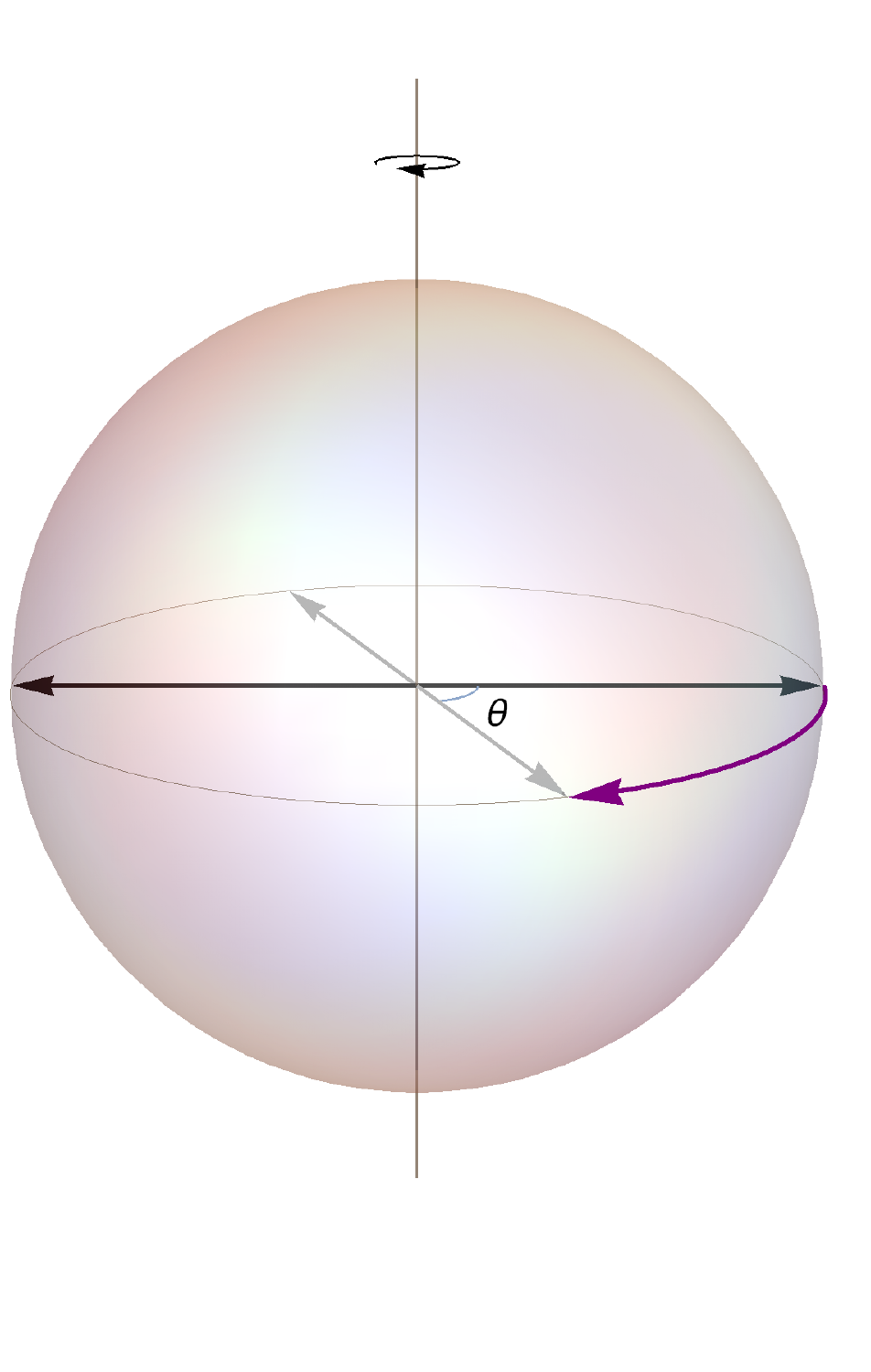}
b) \includegraphics[scale=0.35]{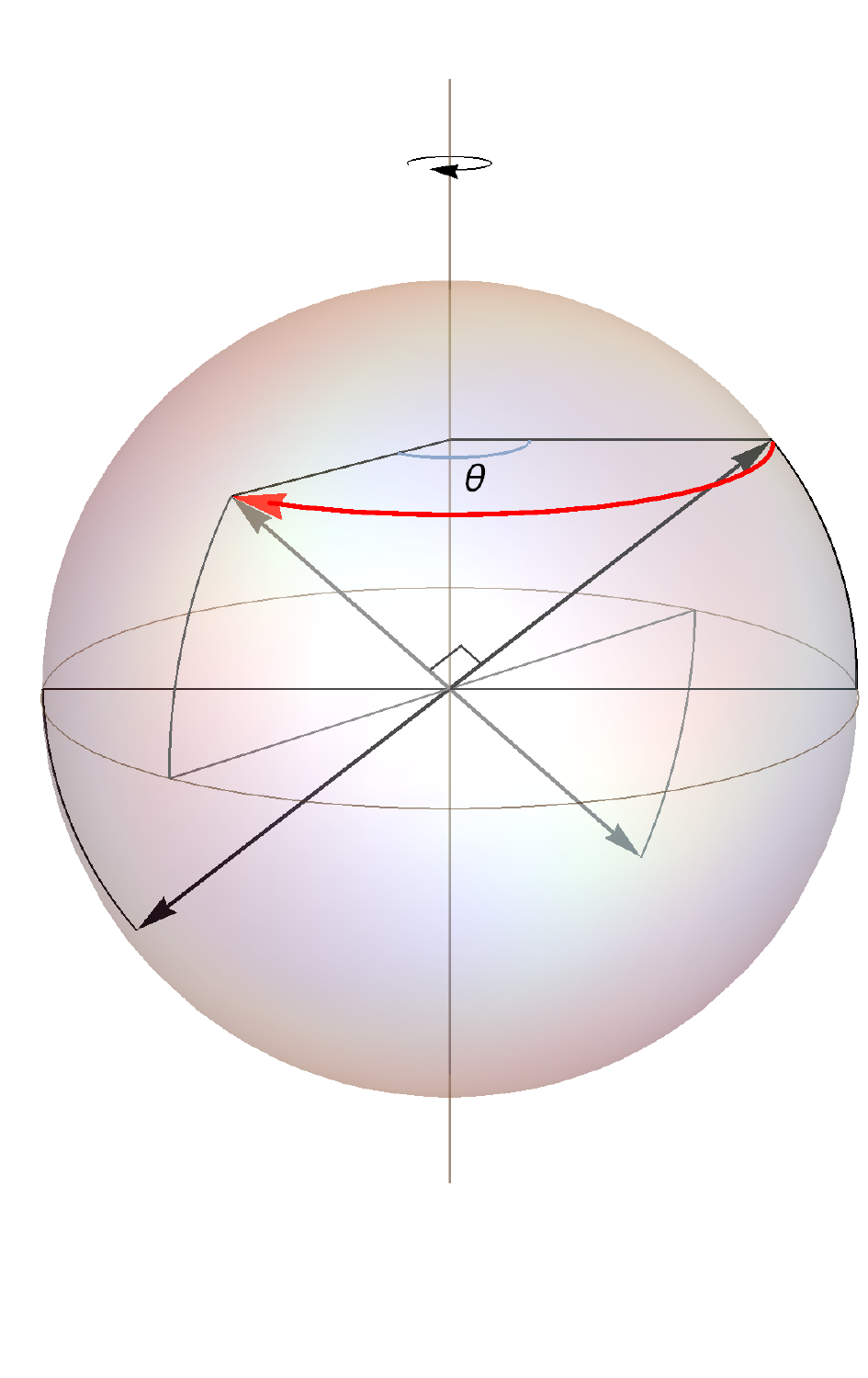}\smallskip

Fig. 2. Maximizers for the PVM-dynamical entropy in dimension $d=2$, where the
unitary map is represented in the Bloch sphere as a rotation by the angle: a)
acute (\textit{purple}) and b) obtuse (\textit{red}).
\end{center}

Next, we compute the volume of the set of chaotic operators in the ensemble of
unitary matrices as well as the average value of the PVM-dynamical entropy. To
this aim we use the Weyl integration formula for $\mathsf{U}(d)$ group
\cite[Theorem 7.4.B]{Wey53}. Recall that $F\colon\mathsf{U}(d)\rightarrow
\mathbb{C}$ is a \textsl{class function} if it is constant on the conjugacy
classes, i.e., for all $U,V\in\mathsf{U}(d)$ we have $F(U)=F(V^{\ast}UV)$.

\begin{theorem*}
[Weyl's integration formula]\label{weyl} If $F\in L^{1}(\mathsf{U}(d))$ is a
class function, then the following formula holds
\[
\int\limits_{\mathsf{U}(d)}F\left(  U\right)  dm\left(  U\right)  =
\]
\[
\frac{1}{d!\left(  2\pi\right)  ^{d}}\int\limits_{[0,2\pi)^{d}}f(\theta
_{1},\ldots,\theta_{d})\prod_{1\leq j<l\leq d}|e^{i\theta_{j}}-e^{i\theta_{l}%
}|^{2}d\theta_{1}\ldots d\theta_{d}\text{,}%
\]
where $m$ denotes the normalized Haar measure on $\mathsf{U}(d)$ and $f(\theta_{1},\ldots,\theta_{d}):=F(\Theta)$ for $\Theta
:={\operatorname*{diag}}(e^{i\theta_{1}},\ldots,e^{i\theta_{d}})$.
\end{theorem*}

Applying this formula, we get

\begin{theorem}
\label{C2}Let $C_{2}:=\left\{  U\in\mathsf{U}(2):U\text{ is chaotic}\right\}
$. Then
\[
m\left(  C_{2}\right)  =\frac{1}{2}+\frac{1}{\pi}\approx0.8183\text{.}%
\]
\end{theorem}

\begin{proof}
It follows from the Weyl integration formula that
\begin{align*}
m\left(  C_{2}\right)   &  =\int_{\mathsf{U}(2)}\mathbf{1}_{C_{2}}(U)dm(U)\\
&  =\frac{1}{4\pi}\int_{\pi/2}^{3\pi/2}|e^{i\varphi}-1|^{2}d\varphi\\
&  =\frac{1}{2\pi}\int_{\pi/2}^{3\pi/2}\left(  1-\cos\varphi\right)
d\varphi=\frac{1}{2}+\frac{1}{\pi}\text{,}%
\end{align*}
as desired.
\end{proof}
We show that the average entropy is in this case not far from its maximal value
${\ln2\approx0.693}$.
\begin{theorem}
The average value of the PVM-dynamical entropy is given by{\
\[
\left\langle H^{\operatorname*{dyn}}(U)\right\rangle _{\mathsf{U}(2)}=\frac
{3}{2}\ln2-\frac{1}{2}-\frac{1}{2\pi}+\frac{C}{\pi}\approx0.672\text{,}%
\]
}where $C$ is Catalan's constant, which may be computed from the formula
\[
C:=\sum_{n=0}^{\infty}\frac{(-1)^{n}}{(2n+1)^{2}}\approx0.916\text{.}%
\]
\end{theorem}

\begin{proof}
Using again the Weyl integration formula and (\ref{fordim2}), we get
\begin{align*}
& \left\langle H^{\operatorname*{dyn}}(U)\right\rangle _{\mathsf{U}(2)}\\
&  =\int_{\mathsf{U}(d)}H^{\operatorname*{dyn}}(U)\ dm\left(  U\right) \\
&  =\frac{1}{\pi}\int_{0}^{\pi/2}\left(  \eta\left(  \cos^{2}\left(
\frac{\varphi}{2}\right)  \right)  +\eta\left(  \sin^{2}\left(  \frac{\varphi
}{2}\right)  \right)  \right)  \left(  1-\cos\varphi\right)  d\varphi \\
& +\left(  \frac{1}{2}+\frac{1}{\pi}\right)  \ln2\text{.}%
\end{align*}
The first summand can be written as the sum of several integrals, which gives
\begin{align}
&  \left\langle H^{\operatorname*{dyn}}(U)\right\rangle _{\mathsf{U}%
(2)}\nonumber\\
&  =\ln2+\frac{1}{\pi}\int_{0}^{\pi/2}\cos\varphi\ln\left(  1-\cos
\varphi\right)  d\varphi\nonumber\\
&  -\frac{1}{2\pi}\int_{0}^{\pi/2}\ln\left(  1+\cos\varphi\right)
\,d\varphi-\frac{1}{2\pi}\int_{0}^{\pi/2}\ln\left(  1-\cos\varphi\right)d\varphi\nonumber\\
&  +\frac{1}{2\pi}\int_{0}^{\pi/2}\cos^{2}\varphi\ln\left(  \frac
{1+\cos\varphi}{1-\cos\varphi}\right)  d\varphi\text{.} \label{a1}%
\end{align}
Firstly, integrating by parts, we get
\begin{equation}
\int_{0}^{\frac{\pi}{2}}\cos\varphi\ln\left(  1-\cos\varphi\right)
d\varphi=-\left(  1+\frac{\pi}{2}\right)  \text{.} \label{a2}%
\end{equation}
\noindent
In the following calculations we use various integral representations of
Catalan's constant, which can be found in \cite{Bra01}. Using the tangent
half-angle substitution $x=\tan\left(  \varphi/2\right)  $ and formula (23)
from \cite{Bra01}, we obtain
\begin{align}
\int_{0}^{\frac{\pi}{2}}\ln\left(  1+\cos\varphi\right)  d\varphi &  =\int
_{0}^{1}\frac{2}{1+x^{2}}\ln\left(  \frac{2}{1+x^{2}}\right)  dx\nonumber\\
&  =\frac{\pi}{2}\ln2-2\int_{0}^{1}\frac{\ln\left(  1+x^{2}\right)  }{1+x^{2}%
}\,dx\nonumber\\
&  =-\frac{\pi}{2}\ln2+2C\text{.} \label{a3}%
\end{align}
\smallskip
From this equality and formula (10) from \cite{Bra01} we get
\begin{equation}
\int_{0}^{\frac{\pi}{2}}\ln\left(  1-\cos\varphi\right)  d\varphi=-\frac{\pi
}{2}\ln2-2C\text{.} \label{a4}%
\end{equation}
Finally, integrating by parts and using formula (4) from \cite{Bra01}, we
have
\smallskip
\begin{equation}
\int_{0}^{\frac{\pi}{2}}\cos^{2}\varphi\ln\left(  \frac{1+\cos\varphi}%
{1-\cos\varphi}\right)  d\varphi=1+\int_{0}^{\frac{\pi}{2}}\frac{\varphi}%
{\sin\varphi}d\varphi=1+2C\text{.} \label{a5}%
\end{equation}
\smallskip
Now, combining (\ref{a1}), (\ref{a2}), (\ref{a3}), (\ref{a4}) and (\ref{a5}),
we obtain
\[
\left\langle H^{\operatorname*{dyn}}(U)\right\rangle _{\mathsf{U}(2)}=\frac
{3}{2}\ln2+\frac{2C-\pi-1}{2\pi}\approx0.672\text{.}%
\]
\end{proof}

\section{PVM-dynamical entropy: qutrits and beyond\label{QUTRITS}}

To determine whether or not a given unitary $U$ belongs to
$C_{d}=\left\{  U\in\mathsf{U}(d):U\text{ is chaotic}\right\}  $, one has to
know its spectrum lying on the unit circle and defined up to a phase factor.
We can, because of this overall phase freedom, restrict our attention to the
set of special unitary matrices and assume that $U\in\mathsf{SU}(d)$. It is
well known that all possible values of the trace of matrices from
$\mathsf{SU}(d)$ fill in the region $T_{d}:=\left\{  \operatorname{tr}%
U:U\in\mathsf{SU}(d)\right\}  $ in the complex plane bounded by a
$d$-hypocycloid with cusps at $d$-th roots of unity scaled up by $d$, i.e.,
the curve produced by a point on the circumference of a small circle of radius
$1$ rolling around the inside of a large circle of radius $d$ and starting at
$(d,0)$ \cite[Theorem 5.2]{Chaetal05}, see also \cite{Kai06}. It follows from
Corollary \ref{trace} that $CT_{d}:=\left\{  \operatorname{tr}U:U\in
\mathsf{SU}(d)\text{, }U\text{ is chaotic}\right\}  $, i.e., the image of the
set of special chaotic matrices under the trace map, is contained in the ball
$B(0,\sqrt{d})$. We shall see that $CT_{d}$ is the subset of $T_{d}\cap
B(0,\sqrt{d})$ (the latter is just $B(0,\sqrt{d})$ for $d\geq4$) given by the union
of regions indexed by pairs consisting of a complex Hadamard matrix of order
$d$ and a permutation of a $d$-element set. Each of these regions is the image
of $T_{d}$ under a spiral similarity with centre at $0$, ratio $1/\sqrt{d}$,
and angle of rotation that depends on the index. Namely, for a given pair
$(H,\sigma)$ we consider the Leibniz formula for the determinant of $H$. A
$d$-th root of the normalized summand in this formula corresponding to $\sigma$
is equal to the complex multiplier defining the spiral similarity.

In fact, it is enough to take here `benchmark' Hadamard matrices defined in the
following way. Denote by $\mathsf{H}_{d}$ the set of all complex Hadamard
matrices of order $d$. We call \mbox{$\mathsf{B}\subset\mathsf{H}_{d}$} a
\textsl{benchmark set} if every $H\in\mathsf{H}_{d}$ is \textsl{equivalent} to
some matrix $F$ in $\mathsf{B}$, i.e., it is of the form $H=D_{1}P_{1}%
FP_{2}D_{2}$, where $D_{1},D_{2}$ are diagonal unitary matrices and
$P_{1},P_{2}$ are permutation matrices. We have

\begin{theorem}
\label{necess}Let $\mathsf{B}\subset\mathsf{H}_{d}$ be a benchmark set. Then
\begin{align}
CT_{d}  &  =
{\textstyle\bigcup}
\left\{  \alpha_{F,\sigma}T_{d}:F\in\mathsf{B},\sigma\in S_{d}\right\}
\nonumber\\
& =
{\textstyle\bigcup}
\left\{  \alpha_{F,\sigma}T_{d}:F\in\mathsf{H}_{d},\sigma\in S_{d}\right\}
\text{,} \label{hyper}%
\end{align}
where for $F\in\mathsf{H}_{d},\sigma\in S_{d}$ we take $\alpha_{F,\sigma}$ to
be any $d$-th root of $\left(  \det F\right)  ^{-1}\left(  \operatorname*{sgn}%
(\sigma)\right)
{\textstyle\prod_{j=1}^{d}}
F_{j,\sigma\left(  j\right)  }$ (and so $|\alpha_{F,\sigma}|=1/\sqrt{d}$).
\end{theorem}

\begin{proof}
Let $U\in\mathsf{SU}(d)\cap C_{d}$. It follows from Proposition \ref{Hadamard}
that $U$ is represented in some orthonormal basis by $H\in\mathsf{H}_{d}$
rescaled by the factor $1/\sqrt{d}$. Fix this basis. Then one can find
$F\in\mathsf{B}$, diagonal unitary matrices $D_{1},D_{2}$ and permutation
matrices $P_{\sigma_{r}}$ corresponding to $\sigma_{r}\in S_{d}$ ($r=1,2$)
such that $H=D_{1}P_{\sigma_{1}}FP_{\sigma_{2}}D_{2}$. Put $D:=D_{2}D_{1}$ and $\sigma
:=\sigma_{2}\circ\sigma_{1}$. Observe that $\operatorname*{sgn}\left(
\sigma\right)  =\det\left( P_{\sigma_{2}}P_{\sigma_{1}}\right)  $. Moreover, $d^{d/2}=\det
H=\det\left(  P_{\sigma_{2}}P_{\sigma_{1}}\right)  \det D\det F$. Let $\lambda_{j}\in\mathbb{C}%
$, $\left|  \lambda_{j}\right| =1$ ($j=1,\ldots,d$) stand for the diagonal
elements of $D$. Set
\[
D^{\prime}:=d^{1/2}\overline{\alpha_{F,\sigma}}\operatorname*{diag}\left(
\lambda_{\sigma_{1}(j)}F_{j,\sigma\left(  j\right)  }\right)  _{j=1}
^{d}\text{.}
\]
Then $D^{\prime}$ is a unitary matrix as $\left|  \alpha_{F,\sigma}\right|
=1/\sqrt{d}$. We have
\begin{align*}
\det D^{\prime}  &  =d^{d/2}(\overline{\alpha_{F,\sigma}})^{d}
{\textstyle\prod_{j=1}^{d}}
\lambda_{\sigma_{1}(j)}F_{j,\sigma\left(  j\right)  }\\
&  =d^{d/2}(\overline{\alpha_{F,\sigma}})^{d}\left(  \det D\right)
{\textstyle\prod_{j=1}^{d}}
F_{j,\sigma\left(  j\right)  }\\
&  =d^{d}(\overline{\alpha_{F,\sigma}})^{d}\left(  \det F\right)
^{-1}\operatorname*{sgn}\left(  \sigma\right)
{\textstyle\prod_{j=1}^{d}}
F_{j,\sigma\left(  j\right)  }\\
&  =d^{d}|\alpha_{F,\sigma}|^{2d}=1\text{.}
\end{align*}
Hence, $\operatorname{tr}D^{\prime}\in T_{d}$. Moreover,
\begin{align*}
\sqrt{d}\operatorname{tr}U  &  =\operatorname{tr}H=\operatorname{tr}%
P_{\sigma_{1}}FP_{\sigma_{2}}D=
{\textstyle\sum_{j=1}^{d}}
\lambda_{j}F_{\sigma_{1}^{-1}(j),\sigma_{2}\left(  j\right)  }
\\
&
 =
{\textstyle\sum_{j=1}^{d}}
\lambda_{\sigma_{1}(j)}F_{j,\sigma(j)}=\sqrt{d}\alpha_{F,\sigma}%
\operatorname{tr}D^{\prime}\text{,}
\end{align*}
and so $\operatorname{tr}U\in\alpha_{F,\sigma}T_{d}$. In this way, we showed
that $CT_{d}\subset
{\textstyle\bigcup}
\left\{  \alpha_{F,\sigma}T_{d}:F\in B,\sigma\in S_{d}\right\}  $.

Now, let $F\in\mathsf{H}_{d}$, $\sigma\in S_{d}$ and $\lambda\in T_{d}$. Then
there is a unitary $U\in\mathsf{SU}(d)$ such that $\operatorname{tr}U=\lambda
$. Fix an eigenbasis of $U$. Then $U$ is represented by a matrix
$\operatorname*{diag}(\kappa_{j})_{j=1}^{d}$, where $\kappa_{j}\in\mathbb{C}$,
$\left|  \kappa_{j}\right|  =1$ ($j=1,\ldots,d$), $
{\textstyle\sum_{j=1}^{d}}
\kappa_{j}=\lambda$ and $
{\textstyle\prod_{j=1}^{d}}
\kappa_{j}=1$. Define $D^{\prime}:=\operatorname*{diag}(\lambda_{j})_{j=1}%
^{d}FP_{\sigma}$ with $\lambda_{j}:=\alpha_{F,\sigma}\kappa_{j}\overline
{F_{j,\sigma\left(  j\right)  }}$ for $j=1,\ldots,d$. Then $\sqrt{d}D^{\prime
}\in\mathsf{H}_{d}$ fulfills
\begin{align*}
\det D^{\prime}  &  =(
{\textstyle\prod_{j=1}^{d}}
\lambda_{j})\operatorname*{sgn}\left(  \sigma\right)  \left(  \det F\right) \\
&  =\alpha_{F,\sigma}^{d}\operatorname*{sgn}\left(  \sigma\right)  \left(
\det F\right)
{\textstyle\prod_{j=1}^{d}}
\overline{F_{j,\sigma\left(  j\right)  }}=1
\end{align*}
and $\operatorname{tr}D^{\prime}=
{\textstyle\sum_{j=1}^{d}}
\lambda_{j}F_{j,\sigma\left(  j\right)  }=\alpha_{F,\sigma}\lambda$. Thus,
$D^{\prime}$ represents, by Proposition \ref{Hadamard}, a chaotic $U^{\prime
}\in\mathsf{SU}(d)$ such that $\operatorname{tr}U^{\prime}=\alpha_{F,\sigma
}\lambda$. Hence, $\alpha_{F,\sigma}\lambda\in CT_{d}$. In consequence, $
{\textstyle\bigcup}
\left\{  \alpha_{F,\sigma}T_{d}:F\in\mathsf{H}_{d},\sigma\in S_{d}\right\}
\subset CT_{d} $, which completes the proof.
\end{proof}

This theorem gives us another characterization of the set of chaotic unitaries
for $d=2$. In this case, since the Fourier matrix $F_{2}$, where
\[
F_{2}=\left[
\begin{array}
[c]{cc}%
1 & 1\\
1 & -1
\end{array}
\right]  \text{,}
\]
serves as the only benchmark Hadamard matrix, we get at once
$CT_{2}=\left\{  x\in\mathbb{R}:\left|  x\right|  \leq\sqrt{2}\right\}
=(1/\sqrt{2})\left\{  x\in\mathbb{R}:\left|  x\right|  \leq2\right\}
=(1/\sqrt{2})T_{2}$. Hence, we obtain the following simple result, which can
also be easily deduced from (\ref{fordim2}).

\begin{proposition}
\label{chaotic2}Let $U\in\mathsf{U}(2)$. Then $U$ is chaotic if and only if
$\left| \operatorname{tr}U\right|  \leq\sqrt{2}$.
\end{proposition}

In the case of qutrits ($d=3$) it follows from Theorem \ref{necess} that
$CT_{3}$, i.e., the image of the set of special chaotic matrices under the
trace map, is the subset of $T_{3}$ given by the union of two regions each of
which is bounded by a $3$-hypocycloid that arises from the original
$3$-hypocycloid (the black curve in Fig. 3) by scaling it down by a factor of
$\sqrt{3}$ and rotating by $\pm\pi/18$ (the union of figures bounded by the
red curves in Fig. 3).

Observe that the characteristic polynomial of $U\in\mathsf{SU}(3)$ takes the
form $\lambda^{3}-\left(  \operatorname{tr}U\right)  \lambda^{2}%
+\overline{\left(  \operatorname{tr}U\right)  }\lambda-1$, so the spectrum of
$U$, and thus the answer to the question whether it is chaotic or not, depends
solely on its trace. Thus, it is not a surprise that in this case the
necessary condition (\ref{hyper}) becomes sufficient as well.
\begin{center}
\includegraphics[scale=0.3]{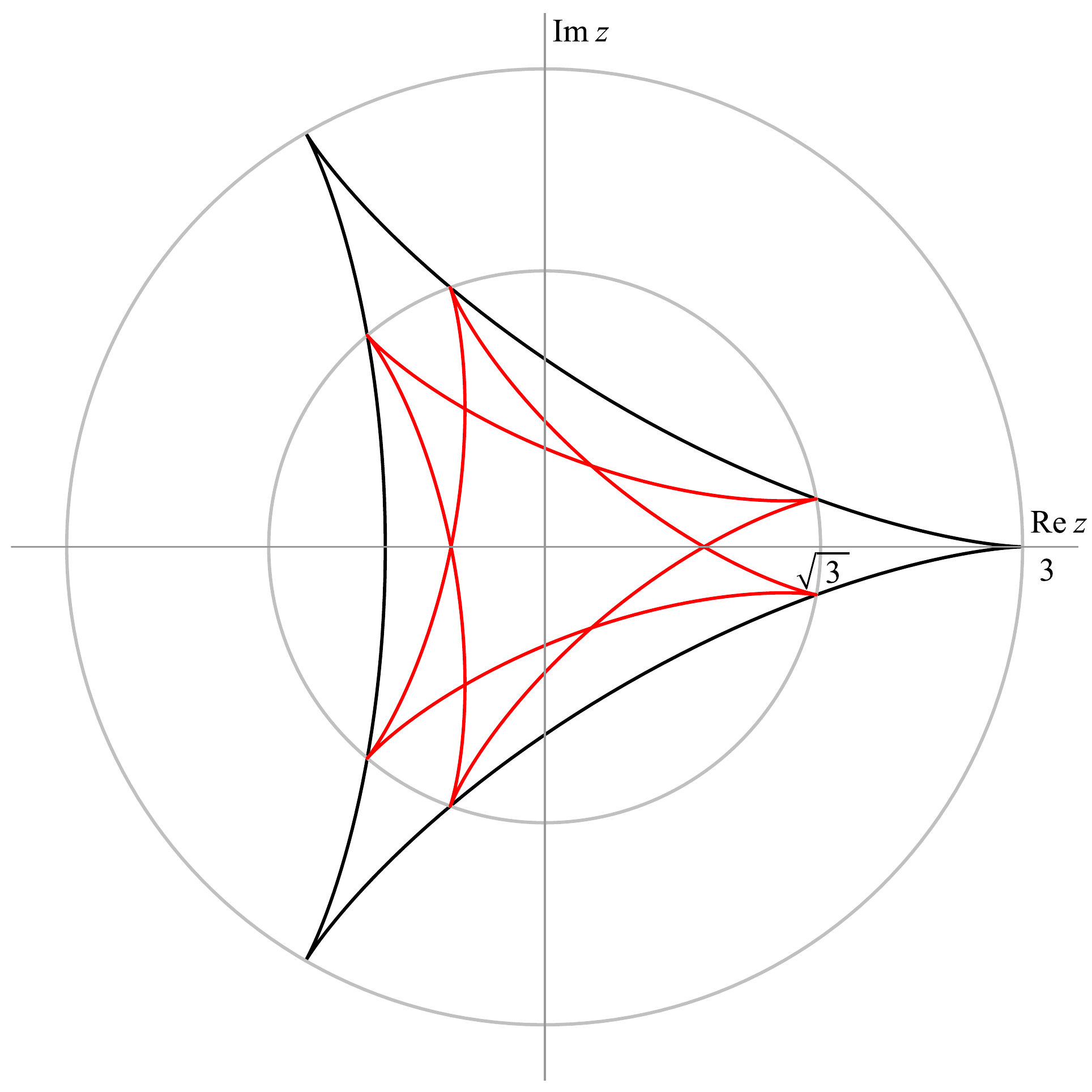}

Fig. 3. Traces of special chaotic unitaries for $d=3$ (the region bounded by
the \textit{red} curves).
\end{center}

\begin{theorem}
\label{chaotic3}Let $U\in\mathsf{U}(3)$ and let $\beta$ be a cube root of
$\det U$. Then $U$ is chaotic iff
\[
\frac{1}{\beta}\operatorname{tr}U\in CT_{3}=\frac{1}{\sqrt{3}}\left(  \alpha
T_{3}\cup\overline{\alpha}T_{3}\right)  \text{,}%
\]
where $\alpha:=e^{\frac{\pi}{18}i}$.
\end{theorem}

\begin{proof}
All complex Hadamard matrices of order $3$ are equivalent to the Fourier
matrix $F_{3}$ \cite{Cra91}, where
\[
F_{3}=%
\begin{bmatrix}
1 & 1 & 1\\
1 & \omega_{3} & \omega_{3}^{2}\\
1 & \omega_{3}^{2} & \omega_{3}%
\end{bmatrix}
\]
and $\omega_{3}:=\exp(2\pi i/3)$. Applying Theorem \ref{necess} and using $\det
F_{3}=-3\sqrt{3}i$, we obtain two possible scaling factors: $\alpha
_{F,\operatorname{id}}^{3}=-\omega_{3}^{2}/(3\sqrt{3}i)=(\overline{\alpha
}/\sqrt{3})^{3}$ and $\alpha_{F,\sigma}^{3}=\omega_{3}/(3\sqrt{3}%
i)=(\alpha/\sqrt{3})^{3}$, with $\sigma\in S_{3}$ defined by $\sigma\left(
1\right)  =1$, $\sigma\left(  2\right)  =3$ and $\sigma\left(  3\right)  =2$,
which implies $CT_{3}=\frac{1}{\sqrt{3}}\left(  \alpha T_{3}\cup
\overline{\alpha}T_{3}\right)  $. Now, the assertion follows from the fact
that the spectrum of $U\in\mathsf{SU}(3)$ is fully defined by its trace.
\end{proof}

Next, we use this result to estimate the volume of the set of chaotic
unitaries in dimension~$3$. First, observe that
\[
m(C_{3})=\mu(U\in\mathsf{SU}(3)\text{, }U\text{ is chaotic})\text{,}%
\]
where $\mu$ stands for the normalized Haar measure on $\mathsf{SU}(3)$. Now,
from the Weyl integration formula for $\mathsf{SU}(3)$
\cite[eq. (9)]{Kai06} and Theorem \ref{chaotic3}, we get

\begin{theorem}
\label{C3}
\[
m(C_{3})=\frac{3\sqrt{3}}{2\pi^{2}}
{\displaystyle\int\limits_{CT_{3}}}
\sqrt{4\negthinspace+\negthinspace\left(  \frac{2r}{3}\right)  ^{3}\cos3\theta\negthinspace-\negthinspace 3\left(  1+\frac{r^{2}%
}{9}\right)  ^{2}}rdrd\theta\text{.}
\]
\end{theorem}
\noindent
Evaluating the above integral numerically, we obtain $m(C_{3})\approx0.592$. It
is noteworthy that $m(C_{3})<m(C_{2})$.

Observe that Theorem \ref{necess} does not provide, however, any new
information about chaotic unitaries for $d=4$, since the one-parameter family
of benchmark Hadamard matrices
\medskip
\[
F_{4}^{(1)}(\varphi):=\left[
\begin{array}
[c]{cccc}%
1 & 1 & 1 & 1\\
1 & ie^{i\varphi} & -1 & -ie^{i\varphi}\\
1 & -1 & 1 & -1\\
1 & -ie^{i\varphi} & -1 & ie^{i\varphi}%
\end{array}
\right]  \text{,}%
\medskip
\]
where $\varphi\in\lbrack0,2\pi)$, see \cite{Cra91} and \cite{TadZyc06}, generates all possible
complex multipliers of modulus $1/2$.

On the other hand, for $d=5$ the benchmark set consists only of the Fourier
matrix%
\medskip
\[
F_{5}:=\left[
\begin{array}
[c]{ccccc}%
1 & 1 & 1 & 1 & 1\\
1 & \omega_{5} & \omega_{5}^{2} & \omega_{5}^{3} & \omega_{5}^{4}\\
1 & \omega_{5}^{2} & \omega_{5}^{4} & \omega_{5} & \omega_{5}^{3}\\
1 & \omega_{5}^{3} & \omega_{5} & \omega_{5}^{4} & \omega_{5}^{2}\\
1 & \omega_{5}^{4} & \omega_{5}^{3} & \omega_{5}^{2} & \omega_{5}%
\end{array}
\right]  \text{,}%
\medskip
\]
where $\omega_{5}:=\exp(2\pi i/5)$, see \cite{Haa96} and \cite{TadZyc06}. By direct calculation
we deduce from Theorem \ref{necess} a simple necessary condition for
$U\in\mathsf{U}(5)$ to be chaotic.

\begin{proposition}
Let $U\in\mathsf{U}(5)$ and $\beta^{5}=\det U$. If $U$ is chaotic, then
\[
\frac{1}{\beta}\operatorname{tr}U\in CT_{5}=\frac{1}{\sqrt{5}}{
{\textstyle\bigcup}
}\left\{  \alpha T_{5}:\alpha\in A\right\}  \text{,}%
\]
where $A:=\{1,-1,e^{\frac{\pi}{25}i},e^{-\frac{\pi}{25}i},e^{\frac{2\pi}{25}%
i},e^{-\frac{2\pi}{25}i}\}$, see Fig. 4.
\end{proposition}

 For higher dimensions ($d\geq6$) Theorem \ref{necess} does not provide
 concrete information about the chaoticity of a unitary map, since the complete
 classification of complex Hadamard matrices is only available up to order
 $d=5$.

\begin{center}
\includegraphics[scale=0.4]{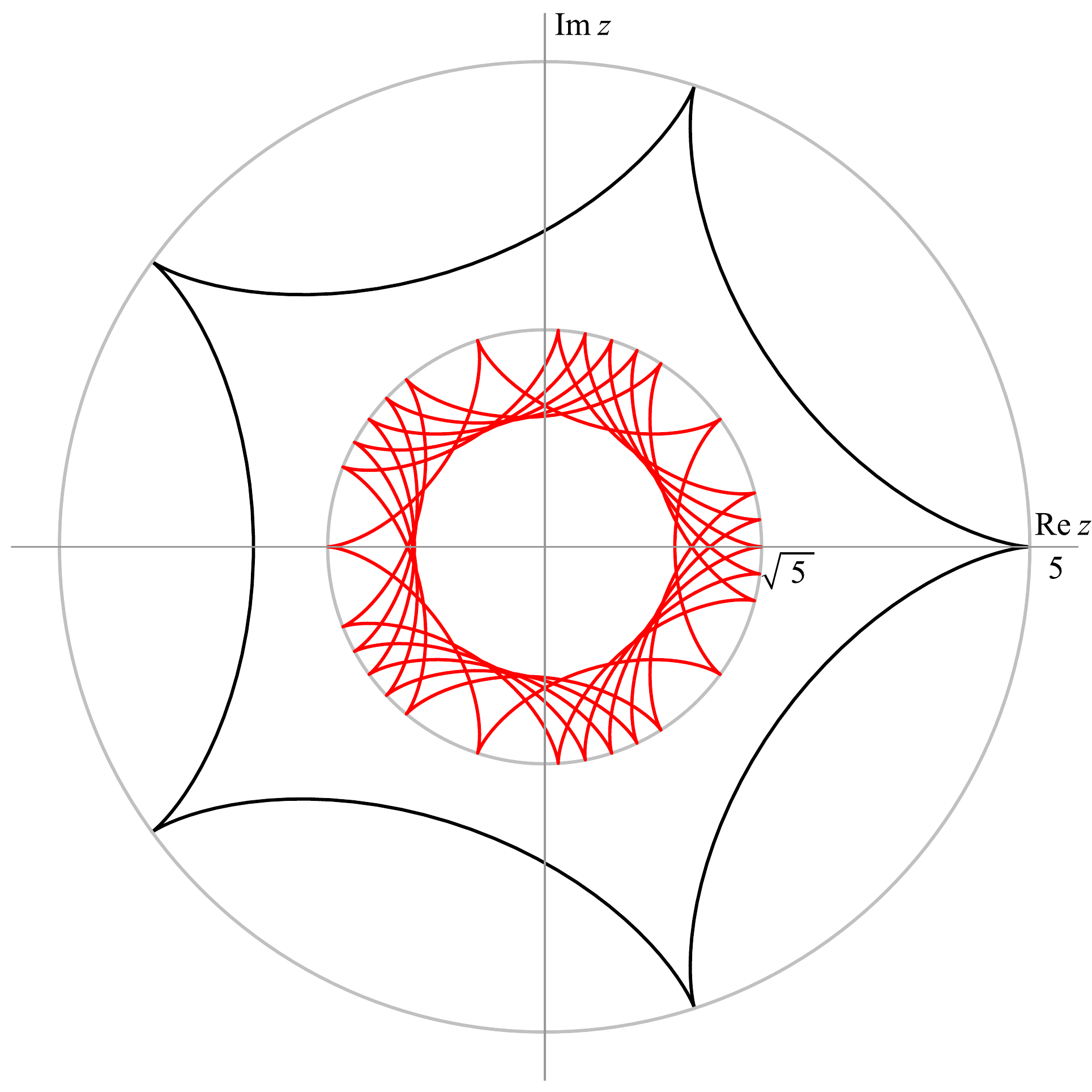}

Fig. 4. Traces of special chaotic unitaries for $d=5$ (the region bounded by
the \textit{red} curves).
\end{center}

\section{Entropy of measurement and POVM-entropy\label{POVM}}

In the closing section we would like to briefly discuss some issues related to
the POVM-dynamical entropy. We start by recalling the notion of entropy of a
POVM. By the \emph{(Shannon) entropy of the measurement} $\Pi=(\Pi
_{j})_{j=1,\ldots,k}$, where $\Pi_{j}=\left(  d/k\right)  \left|  \varphi
_{j}\right\rangle \left\langle \varphi_{j}\right|  $ for $\left|  \varphi
_{j}\right\rangle \left\langle \varphi_{j}\right|  \in\mathcal{P}\left(
\mathbb{C}^{d}\right)  $ ($j=1,\ldots,k$), we mean the function $H(\cdot
,\Pi)\colon\mathcal{S}\left(  \mathbb{C}^{d}\right)  \rightarrow\mathbb{R}$
defined by
\begin{align*}
H(\rho,\Pi)  &  :=\sum_{j=1}^{k}\eta(p_{j}(\rho,\Pi))\\
&  =\sum_{j=1}^{k}\eta(\left(  d/k\right)  \left\langle \varphi_{j}\right|
\rho\left|  \varphi_{j}\right\rangle )\\
&  =\ln\frac{k}{d}+\frac{d}{k}\sum_{j=1}^{k}\eta(\left\langle \varphi
_{j}\right|  \rho\left|  \varphi_{j}\right\rangle )
\end{align*}
for an input state $\rho\in\mathcal{S}\left(  \mathbb{C}^{d}\right)  $; see
\cite{SloSzy16, Wil16} for the history and information-theoretic
interpretation of this notion. If $\rho=\left|  \psi\right\rangle \left\langle
\psi\right|  \in\mathcal{P}\left(  \mathbb{C}^{d}\right)  $, we put $H(\left|
\psi\right\rangle ,\Pi):=H(\rho,\Pi)$. Applying (\ref{formula}), we see that
the entropy of $U\in\mathsf{U}(d)$ with respect to $\Pi$ can be expressed as
the mean entropy of $\Pi$ averaged over the output states of $\Pi$ transformed
by $U$:
\begin{equation}
H\left(  U,\Pi\right)  =\frac{1}{k}\sum_{j=1}^{k}H(U\left|  \varphi
_{j}\right\rangle ,\Pi) \label{dynentrel}%
\end{equation}
and so
\begin{equation}
H_{\operatorname*{meas}}\left(  \Pi\right)  =H\left(  \mathbb{I},\Pi\right)
=\frac{1}{k}\sum_{j=1}^{k}H(\left|  \varphi_{j}\right\rangle ,\Pi)\text{.}
\label{meaentrel}%
\end{equation}
From (\ref{dynentrel}) and (\ref{meaentrel}) we obtain
\begin{equation}
H_{\operatorname*{dyn}}\left(  U,\Pi\right)  =\frac{1}{k}\sum_{j=1}%
^{k}[H(U\left|  \varphi_{j}\right\rangle ,\Pi)-H(\left|  \varphi
_{j}\right\rangle ,\Pi)]\text{.} \label{dynfor}%
\end{equation}

For PVMs we get $H(\mathbb{I},\Pi)=0\leq H(U,\Pi)=H_{\operatorname*{dyn}%
}\left(  U,\Pi\right)  $. Surprisingly, in the general case we can find
situations where intertwining a POVM-measurement with some (or even any)
unitary operator can produce smaller entropy than that generated by the
measurement itself.

To illustrate this phenomenon, consider a SIC-POVM $\Pi=(\Pi_{j})_{j=1,\ldots,d^{2}}$, 
i.e., a rank-1 POVM satisfying the condition
$\operatorname*{tr}(\Pi_{j}\Pi_{l})=1/(d^{2}(d+1))$ for $j,l=1,\ldots,d^{2}$,
$j\neq l$. Then, from (\ref{dynfor}) and \cite{Szy16}, we get
\[
H(U,\Pi)\leq H(\mathbb{I},\Pi)=\frac{d-1}{d}\ln(d+1)+\ln d\text{.}%
\]
We also have (see \cite{Dal15,SzySlo16}) the following bound
\[
\ln\frac{(d+1)d}{2}\leq H(U,\Pi)\text{,}%
\]
which is known to be actually attained for the `tetrahedral' SIC-POVM in
dimension $2$ \cite{SloSzy16}, for all SIC-POVMs in dimension $3$
\cite{Szy14} as well as for the Hoggar SIC-POVM in dimension $8$ \cite{SzySlo16}.
Consequently, for every $U\in\mathsf{U}(d)$ we get
\begin{equation}
-\ln2+\frac{\ln(d+1)}{d}\leq H_{\operatorname*{dyn}}\left(  U,\Pi\right)
\leq0\text{.} \label{entsic}%
\end{equation}
Thus, from this point of view, SIC-POVMs and PVMs lie on the opposite ends of
the spectrum. It seems that the interplay between the two kinds of randomness,
one coming from the measurement and one associated with unitary evolution,
makes the study of the POVM-dynamical entropy particularly difficult.

\section{Conclusions}

In the present paper we solve some problems concerning chaotic unitaries and
PVM-dynamical entropy; however, many questions remain unanswered. We get
several sufficient and/or necessary conditions for a matrix to maximize the
PVM-dynamical entropy (Proposition \ref{Hadamard}, Corollary \ref{trace},
Theorem \ref{necess}), but only for qubits (Propositions \ref{entdyn2} and
\ref{chaotic2}) and qutrits (Theorem \ref{chaotic3}) we can fully describe the
set of chaotic unitaries. The problem of characterising this property in
higher dimensions, starting from ququads, remains open. The fact that the
probability of finding a chaotic matrix among unitaries is smaller for qutrits
(Theorem \ref{C3}) than for qubits (Theorem \ref{C2}) suggests the conjecture
that this probability decreases, possibly to zero, when the dimension of the
Hilbert space grows to infinity. This contrasts with the result of Theorem
\ref{meaent} that the mean PVM-dynamical entropy increases logarithmically
with the dimension and is, in fact, almost as large as possible.

On the other hand, extending the definition of the dynamical entropy to other
classes of measurements opens up a number of natural questions for further
analysis. Moving on to a broader class of POVMs, we are faced, especially in
the case of SIC-POVMs (eq. (\ref{entsic})), with the paradoxical fact that a
unitary dynamics suitably combined with a measurement can decrease the randomness,
which provides an example of a phenomenon with no classical counterpart. It is
also not clear whether the inequality between the PVM-dynamical entropy and
the POVM-dynamical entropy can be sharp.

Leaving the realm of rank-1 operators, we encounter an even more interesting
situation both in the case of PVMs and that of POVMs, first described by one of us (W.S.)
in the more general setting of operational approach to (quantum) dynamics
and measurement process \cite{Slo03}, and then by Wiesner and co-authors in
a series of papers \cite{CruWie08a,CruWie08,CruWie10,Wie10,Monetal11}. In this case
the measurement process together with the unitary dynamics still produces a Markov
chain in the space of states, but the accompanying process generated in the
space of the measurement outcomes does not have to be Markovian (it was first
noted in \cite{BecGra92}), which makes computing the dynamical entropy more challenging.
A detailed discussion of this situation is postponed to a separate publication.

Finally, a natural direction for further research is to study the
semiclassical limit of the dynamical entropies defined here, see also
\cite{SloZyc98}.

\section*{Acknowledgments}
We are thankful to Robert Craigen, S\l awomir Cynk, Zbigniew Pucha\l a, Anna Szymusiak,
and Karol \.{Z}yczkowski for helpful remarks. Financial support by the Polish
National Science Center under Project \mbox{No. DEC-2015/18/A/ST2/00274} is
gratefully acknowledged.

\begin{IEEEbiographynophoto}{Wojciech S\l omczy\'{n}ski}
received all his degrees: the M.Sc. (1984), the Ph.D. (1991), and the habilitation (2004) in mathematics from the Jagiellonian University, Krak\'{o}w, Poland. He is currently an adjunct professor in the Institute of Mathematics and the chairman of the academic board of the Center for Quantitative Research in Political Science at the Jagiellonian University. His research interests include dynamical systems (in particular: chaos, entropy, and fractals), as well as application of mathematics in quantum information and social choice theories. Together with a physicist Karol \.{Z}yczkowski, he proposed in 2004 an alternative voting system for the Council of the European Union, known as the ‘Jagiellonian Compromise’, and in 2011 was a member of a group of mathematicians and social scientists that prepared a new apportionment scheme for the European Parliament called the ‘Cambridge Compromise’. His recent research concerns quantum dynamical entropy, as well as the geometric configurations in the quantum state space extremizing the entropy of quantum measurements.
\end{IEEEbiographynophoto}
\begin{IEEEbiographynophoto}{Anna Szczepanek}
received the M.Sc. degrees in financial mathematics and in applied mathematics, in 2013 and 2014, respectively, both from the Jagiellonian University, Krak\'{o}w, Poland. She is currently a PhD candidate in the Department of Mathematics and Computer Science, under the supervision of Wojciech S\l omczy\'{n}ski, and a research assistant in the Institute of Mathematics. Her main research interests include quantum information theory, Markov processes, and social networks.
\end{IEEEbiographynophoto}

\vspace{5cm}

\end{document}